\documentclass[letterpaper, 10 pt, conference]{ieeeconf}
\usepackage{comment}
\usepackage{amssymb, amsmath,graphicx,charter, latexsym}
\usepackage{amsfonts}

\newtheorem{definition}{Definition}
\newtheorem{lemma}{Lemma}

\newtheorem{theorem}{Theorem}

\usepackage{graphicx}
\usepackage{url}
\usepackage{epstopdf}

\IEEEoverridecommandlockouts                              
\def\bs{\ensuremath\boldsymbol}
\begin{document}
\title{Optimizing Quality of Experience of Dynamic Video Streaming over Fading Wireless Networks
}

\author{Rahul Singh$^{1}$ and P. R. Kumar$^{1}$ 
\thanks{This material is based upon work partially supported by AFOSR Contract FA9550-13-1-0008,  and NSF under Contract Nos. CNS-1302182 and Science \& Technology Center Grant CCF-0939370.}
\thanks{$^{1}$Rahul Singh, and P. R. Kumar are with Texas A\&M University, College Station, TX 77840, USA.
        {\tt\small rsingh1@tamu.edu, prk@tamu.edu}}%
}

\maketitle
\IEEEpeerreviewmaketitle

\begin{abstract}
We address the problem of video streaming packets from an Access Point (AP) to multiple clients over a shared wireless channel with fading. In such systems, each client maintains a buffer of packets from which to play the video, and an outage occurs in the streaming whenever the buffer is empty. Clients can switch to a lower-quality of video packet, or request packet transmission at a higher energy level in order to minimize the number of outages plus the number of outage periods and the number of low quality video packets streamed, while there is an average power constraint on the AP.
We pose the problem of choosing the video quality and transmission power as a Constrained Markov Decision Process (CMDP). We show that the problem involving $N$ clients decomposes into $N$ MDPs, each involving only a single client, and furthermore that the optimal policy has a threshold structure, in which the decision to choose the video-quality and power-level of transmission depends solely on the buffer-level.   
\end{abstract}
\section{Introduction}
Scheduling packets for video streaming over a shared wireless downlink is of increasing attention~\cite{cisco}. Predominantly, this problem has been addressed with the goal of minimizing the average number of outages, i.e., time-slots during which a client has no packet to play~\cite{a2,a3,a1}, \cite{a4},\cite{a5,a6,a7,a8,a9}. However the models considered in these works do not incorporate the communication constraints imposed by the network over which the streaming occurs. Typically clients streaming video files will share a common wireless channel, which again typically has a constraint on the average power. The access point (AP) has to choose the power  level at which to transmit individual packets to each client so as to maximize the total Quality of Experience (QoE) experienced by the clients. The system also has an additional degree of freedom in that the AP can transmit lower quality packets on occasion, leading to a softer loss of video quality than an abrupt outage. 
Another important aspect is that the quality of video streaming experienced by a client depends not only on the number of outages, but also on the number of ``outage periods", i.e., number of interruption periods as well. Thus an outage lasting $10$ time-slots is not the same as $10$ outages each lasting $1$ time-slot.
The QoE experienced by a client thus has to take into account several metrics: the average number of outages, the average number of outage-periods, and the quality of video-packets streamed. In this paper we address this overall problem. While we focus here on the single last-hop case for ease of exposition and brevity, our results can be generalized to multi-hop networks as well. In order to provide non-interruptive video streaming experience to the clients, the AP has to guarantee some sort of service regularity to the clients, i.e., it has to ensure that the packet deliveries to the clients are not in a bursty fashion. References \cite{rahul,rahul1,Rahul2015,guosingh,rs} develop a framework to design policies which provide services to clients in a regular fashion, though not in a video streaming context. 
\section{System Description}
Consider a system where a wireless channel is shared by $N$ clients for the purpose of streaming video packets.\begin{figure}[!t]
	\centering
	\includegraphics[width=0.5\textwidth]{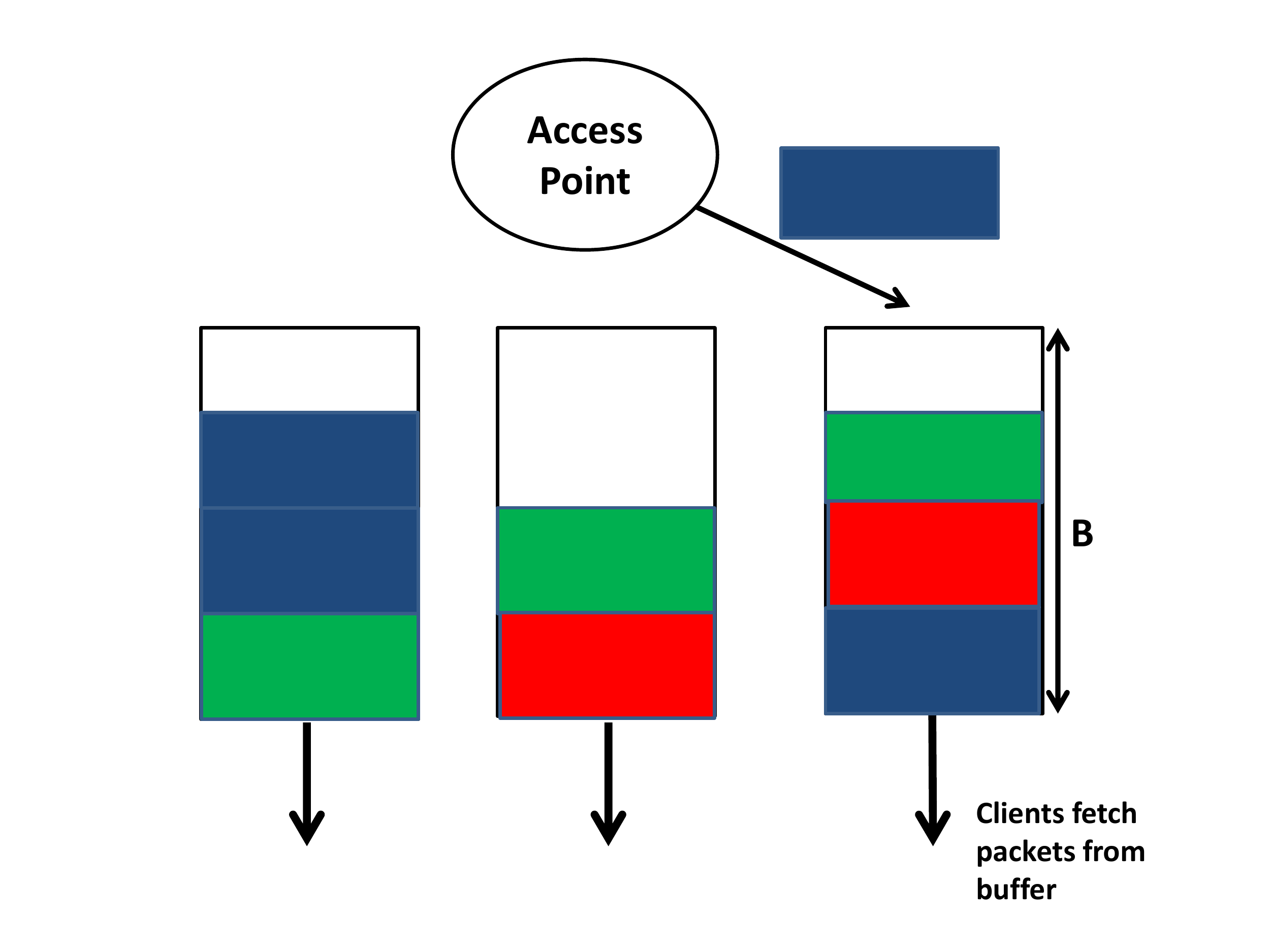}
	\caption{Clients video streaming packets from an Access Point over a shared wireless channel. $B$ denotes the buffer-size, while different colours denote packets of different video qualities.}
	\label{fig1}
\end{figure}
 It is assumed that the system evolves over discrete time-slots, and one time-slot is taken by the access point (AP) for attempting one packet transmission. 

Client $n$ maintains a buffer of size $B_n$ packets and plays a packet for a duration of $T_n$ time-slots. Once it has finished playing a video packet, it looks for the next packet in the buffer. In case the buffer is empty, there is an ``outage", meaning that the video streaming is interrupted, and the client has to wait for a packet to be delivered to its buffer before it can resume the video streaming. 

The wireless channels connecting the clients to the AP are assumed to be random. For ease of exposition, we will derive the results for the case when the channel conditions are fixed. These results carry over to the case of fading channels in a straight-forward manner. Later, in Section~\ref{fading}, we will outline the results for the case of fading channels. 

There are $Q_n$ different video qualities $\{1,2,\ldots,Q_n\}$ of packets that can be transmitted for client $n$, with class $1$ video quality providing the best viewing experience. Similarly there are $\{\hat{E}_1,\hat{E}_2,\ldots,\hat{E}_n\}$ different power levels at which the packets for client $n$ can be transmitted. We let $\hat{E}_1=0$, i.e., a user may choose to not request a packet in a time-slot. The probability that the packet for client $n$ is successfully delivered upon a transmission attempt, $P_n(q,E)$, depends on the amount of power $E$ used in the packet transmission and the quality of video packet $q$ that was attempted. We also incorporate an average power constraint on the AP. 

The basic problem considered is that of scheduling the AP's packet transmissions to clients so as to maximize the combined Quality of Experience (QoE) of the clients. The QoE of a single client depends on multiple factors
\begin{enumerate}
\item The average number of outages.
\item How ``often" the video gets interrupted, i.e., the number of outage-periods, or the number of time-slots in which the transition from ``non-outage" to outage occurs. 
\item The number of packets of different quality types that are streamed.
\end{enumerate}
\section{Problem Formulation}
We denote by $O_n(s)$ the random variable that assumes the value $1$ if the $n$-th client faces an outage at time $s$, and $0$ otherwise, and by $E_n(s)$ the transmission power utilized by the $n$-th client at time-slot $s$. Also, let $I_n(q,s)$ be the random variable that takes the value $1$ if a packet of quality $q$ is delivered to client $n$ in time-slot $s$.

The Constrained Markov Decision Process (CMDP) of interest is then to choose the quality of video packets and transmission power for each client, in order to 
\begin{align}\label{pmdp}
&\mbox{Minimize} \limsup_{t \to \infty}\frac{1}{t}\mathbb{E}\sum_n\sum_{s}\left( O_n(s)+\sum_{q=1}^{Q_n}\lambda_{q,n} I_n(q,s) \right. \notag\\
&\qquad \left.  \vphantom{\left( O_n(s)+\sum_{q=1}^{Q_n}\lambda_{q,n} I_n(q,s) \right)}+ \lambda_{O,n} \lvert O_n(s)\left(O_n(s-1)-1\right)\rvert \right)\notag\\
& \qquad\mbox{ subject to },\\
&\limsup_{t\to\infty}\frac{1}{t}\mathbb{E}\sum_n\sum_{s}E_n(s) \leq \bar{E}.
\tag{Primal MDP}
\end{align}

Note that the term $|O_n(s)\left(O_n(s-1)-1\right)| $ assumes the value $1$ if time-slot $s$ is the beginning of an outage-period for client $n$, and is $0$ otherwise. It thereby measures the number of outage periods incurred. The parameters $\left\{ \lambda_{q,n}\right\}_{q=1}^{Q_n},\lambda_{O,n}\quad n=1,2,\ldots,N$ are employed for tuning the QoS to account for the relative importance placed on each of the objectives. We note that for $i>j$, $\lambda_{i,n}> \lambda_{j,n}$ for all $n$, since we assumed that the video quality of a packet is less if the packet belongs to a higher valued class.
 
Thus the above problem is a CMDP in which the system state at time $t$ is described by the $N$ dimensional vector $L(t):=\left(l_1(t),l_2(t),\ldots,l_N(t)\right)$, where $l_n(t)$ is the amount of play time remaining in the buffer of client $n$ at time $t$.

The central difficulty which arises is that the cardinality of the state-space of the system increases exponentially with the number of clients $N$, and thus the problem is computationally infeasible as formulated above. 

We show that the problem of serving $N$ clients can be decomposed into $N$ separate problems each involving only a single client. Thus the computational complexity of the problem grows linearly in the number of clients. Moreover, we show that the optimal policy is easily implementable since it has a simple threshold structure.   
\section{The Dual MDP}
The Lagrangian associated with a policy $\pi$ for the system~\eqref{pmdp} is given by,
\begin{align}\label{lagrange}
&\mathcal{L}(\pi,\lambda_E) = \limsup_{t \to \infty}\frac{1}{t}\mathbb{E}\sum_n\sum_{s}\left( O_n(s)+\sum_{q=1}^{Q_n}\lambda_{q,n} I_n(q,s)\right. \notag\\
&\qquad\left.  \vphantom{\left( O_n(s)+\sum_{q=1}^{Q_n}\lambda_{q,n} I_n(q,s) \right)} + \lambda_{O,n} |O_n(s)\left(O_n(s-1)-1\right)|  \right)\notag\\
& + \lambda_{E}\left(\limsup_{t\to\infty}\frac{1}{t}\mathbb{E}\sum_n\sum_{s}E_n(s)-\bar{E}\right),
\end{align}
where $\lambda_E$ is the Lagrangian multiplier associated with the average power constraint.
The associated Lagrange dual is,
\begin{align}\label{dual}
D(\lambda_E)= \min_{\pi}\mathcal{L}(\pi,\lambda_E).
\end{align}
Next we present a useful bound on the dual, the proof of which follows from the super-additivity of $\lim\sup$ and sub-additivity of $\lim\inf$ operations.
\begin{lemma}\label{l1}
\begin{align}\label{bound}
&D(\lambda_E) \geq \min_{\pi} \sum_n \liminf_{t \to \infty} \frac{1}{t}\mathbb{E}\sum_{s=1}^{t} \left(  \vphantom{\left( O_n(s)+\sum_{q=1}^{Q_n}\lambda_{q,n} I_n(q,s) \right)} O_n(s)+ \lambda_{E}E_n(s) \right.\notag\\
&\left.+\lambda_{O,n} |O_n(s)\left(O_n(s-1)-1\right)| +\sum_{q=1}^{Q_n} \lambda_{q,n}I_n(q,s) \right)\notag\\
&-\lambda_{E}\bar{E}.
\end{align}
\end{lemma}
\section{Single Client Problem} \label{scp}
We consider minimizing the bound obtained in Lemma~\ref{l1}. Observing the bound, we find that we have decomposed the original problem~\eqref{pmdp} into $N$ single-client problems, i.e., the expression in the r.h.s. of~\eqref{bound} is the sum of the costs of $N$ clients, in which the cost of a single client depends only on the action chosen for it in each time-slot. 

The problem for the single client is described as follows. We omit the sub-script $n$ in the following discussion. The channel connecting the client to the AP is random. The client maintains a buffer of capacity $B$ time-slots of play-time video (this assumption is equivalent to the assumption of maintaining a buffer of $B$ packets since a packet is played for $T$ time-slots), and in each time-slot, the AP has to choose two quantities, which together comprise the control action chosen for the client:
\begin{itemize}
\item The video quality $ q \in \{1,2,\ldots,Q\}$. 
\item The power $E\in \{\hat{E}_1,\hat{E}_2,\ldots,\hat{E}_n\}$ at which to carry out packet transmission.
\end{itemize}
The state of the client is thus described by $l(t)$, the play-time duration of the packets present in the buffer at time $t$. If the client is scheduled a packet transmission of quality $q$ at an power $E$ at time $t$, and the remaining playtime at time $t$, $l(t)$, is less than or equal to $ B-T+1$, then the system state at time $t+1$ is $(l(t)-1)^{+}+T$ with a probability $P(q,E)$, while it is $(l(t)-1)^{+}$ with a probability $P(q,E)$. However if the value of remaining playtime $l(t)$ is strictly greater than $B-T+1$, then the system state at time $t+1$ is $l(t)-1$ with a probability $1$.

We let  
\begin{align}\label{sf}
&\mathcal{S}(x) := 
\begin{cases}
(x-1)^{+}+T,\mbox{ if } x \leq B-T+1,\\
x-1,\mbox{ if } B-T+1<x\leq B, 
\end{cases}\\
&\mathcal{F}(x) := (x-1)^{+},
\end{align}
be the transitions associated with the remaining play-times associated for a successful and failed packet transmission respectively. The control action at time $t$ will be denoted $\boldsymbol{u}(t):=(q(t),E(t))$, where $q(t),E(t)$ are the video quality and transmission power level chosen at time $t$.

The transmissions at power level $E$ incur a cost of $\lambda_{E} \times E$. There is a penalty of $1$ unit upon an outage at time $t$. A penalty of $\lambda_{q}$ units is imposed if a packet of quality $q$ is delivered to it, while a penalty of $\lambda_{O}$ units is imposed at time $t$ in case there was no outage at time-slot $t-1$, and an outage occurs in time-slot $t$, i.e. if a new outage-period begins at time $t$.

Since the probability distribution of the system state at time $t+1$ is completely determined by the system state at time $t$, and the action $\left(q,E\right)$ chosen at time $t$, i.e., requested video quality and power level at which transmission occurs, the single client problem is a Markov Decision Process (MDP) involving only a finite number of actions and states, and is thus solved by a stationary Markov policy~\cite{a12}. 

Denote by $\pi_n$ a policy for the client $n$. The single client problem is to solve,
\begin{align}
&\min_{\pi} \liminf_{t \to \infty} \frac{1}{t}\mathbb{E}\sum_{s=1}^{t} \left(\vphantom{lambda_{O} |O(s)\left(O(s-1)-1\right)| +\sum_l \lambda_{l}I(l,s)} O(s)+ \lambda_{E}E(s) \right.\notag\\
&\left.+\lambda_{O} |O(s)\left(O(s-1)-1\right)| +\sum_{q=1}^{Q} \lambda_{q}I(q,s) \right).
\end{align}
Denote by $\pi^{\star}_n(\lambda_E)$, the optimal policy which solves the single client problem. 
We also let
\begin{align}\label{scopol}
V_n(\lambda_E) = &\min_{\pi} \liminf_{t \to \infty} \frac{1}{t}\mathbb{E}\sum_{s=1}^{t} \left( \vphantom{lambda_{O} |O(s)\left(O(s-1)-1\right)| +\sum_l \lambda_{l}I(l,s)} O(s)+ \lambda_{E}E(s) \right.\notag\\
&\left.+\lambda_{O} |O(s)\left(O(s-1)-1\right)| +\sum_{q=1}^{Q} \lambda_{q}I(q,s) \right),
\end{align}
be the optimal cost, and $V_n(\lambda_E,\pi)$ be the cost associated with a policy $\pi$.
\section{Threshold Structure of the Optimal Policy for the Single Client Problem}
We will suppress the subscript $n$ in the following discussion, and begin with a discussion of the $\beta\in\left(0,1\right)$ discounted infinite horizon cost problem for the single client. Let
\begin{align}\label{disc}
& V_\beta(x) = \min_{\pi}\liminf_{t\to\infty}\mathbb{E}\left[ \sum_{t=0}^{\infty}\beta^t\left(O(t)+\lambda_E E(t)\right.\right.\notag\\ 
& \left.\left. \qquad+\lambda_O |O(t)\left(O(t-1)-1\right)|+\sum_{q=1}^{Q}\lambda_q I(q,s) \right)\right]
\end{align} 
be the minimum $\beta$-discounted infinite horizon cost for the system starting in state $x$ at time $0$, where $x$ can assume values in the set $\{0,1,\ldots,B\}$. The function $V^s_\beta(x)$ is similarly defined to be the minimum $\beta$-discounted cost incurred in $s$ time-slots for the system starting in state $x$, i.e.,
\begin{align*}
& V^s_\beta(x) = \min_{\pi^s}\mathbb{E}_{x}\left[ \sum_{t=0}^{s}\beta^t\left(O(t)+\lambda_E E(t)\right.\right.\\ 
& \left.\left. \qquad+\lambda_O |O(t)\left(O(t-1)-1\right)|+\sum_{q=1}^{Q} \lambda_q I(q,s) \right)\right],
\end{align*} 
where $\pi^s$ is a policy for the $s$ horizon $\beta$-discounted problem. The quantities $ V_\beta(x),V^s_\beta(x)$ should not be confused with the quantities $V_n(\lambda_E)$ defined in the previous section.
We have,
\begin{align}\label{eq:1}
&V^{s}_\beta(x)  = \min_{(q,E)}1(x=0)+ \lambda_E E\notag \\
&+P(q,E)\left[\lambda_q+\beta V^{s-1}_\beta(\mathcal{S}(x)) \right]\notag\\
&+\left(1-P(q,E)\right)\left[1(x=1)\lambda_O+\beta V^{s-1}_\beta(\mathcal{F}(x)) \right]\notag\\
& =1(x=0) + 1(x=1)\lambda_O+ \left[\beta V^{s-1}_\beta(\mathcal{F}(x)) \right]\notag\\
&+ \min_{\boldsymbol{u}}  \{C(\boldsymbol{u})-P(\boldsymbol{u}) D^\beta_s(x)\},
\end{align}
where
\begin{align}\label{cost}
C(\boldsymbol{u}):=\lambda_E E + P(q,E)\lambda_q,
\end{align}
is the one-step cost associated with the action $\boldsymbol{u}=(q,E)$, and for $s=1,2,\ldots$,
\begin{align}\label{df}
D^\beta_{s}(x):=1(x=1)\lambda_O + \beta\left\{ V_\beta^{s-1}(\mathcal{F}(x))-V_\beta^{s-1}(\mathcal{S}(x))\right\}. 
 \end{align}
We assume that a lower video quality packet, or a higher power packet transmission, leads to an increase in the success of packet transmission $P(q,E)$, i.e., an increase in cost is associated with a higher transmission success probability. 
\begin{definition}\notag
We say a policy is of threshold-type if it satisfies the following for each stage $s$:
\begin{itemize}
\item Fix any $E \in \{\hat{E}_1,\hat{E}_2,\ldots,\hat{E}_n\}$. If the policy chooses the action $\left(q,E\right)$ in state $x$, then it does not choose the actions $\{\left(\hat{q},E\right): \hat{q} < q \}$ for any state $1\leq y\leq x$.             \item Fix any $q \in \{Q_1,Q_2,\ldots,Q_n\}$. If the policy chooses the action $\left(q,E\right)$ in state $x$, then it does not choose the actions $\{(q,\tilde{E}): \tilde{E} < E \}$ for any state $1\leq y\leq x$.             \end{itemize}
\end{definition}

If $x,y \in \{1,2,\ldots,B\}$ are such that $x>y$, let $\bs{u}_x,\bs{u}_y$ be the actions chosen by a threshold policy $\pi$ in states $x$ and $y$. Then it is easily verified that $P(\bs{u}_x)<P(\bs{u}_y)$.

Next we present a useful lemma that is easily proved. In the following, $(\bs{u},\pi)$ is the policy that follows the action $\bs{u}$ in the first slot, and then follows policy $\pi$, while $V^{s,\pi}_\beta(x)$ is the cost achieved under the policy $\pi$ in $s$ time-slots for the system starting in state $x$. 
\begin{lemma}\label{l2}
Let $\bs{u}_1,\bs{u}_2$ be two actions where $P(\bs{u}_2)>P(\bs{u}_1)$, or equivalently, $P(\bs{u}_2)>P(\bs{u}_1)$. Then,
\begin{align*}
& V^{s,(\bs{u}_2,\pi^{\star})}_\beta(\mathcal{F}(x))-V^{s,(\bs{u}_1,\pi^{\star})}_\beta(\mathcal{S}(x)) = \\
& P(\bs{u}_1) \left\{\beta V^{s-1}_\beta(\mathcal{S}(\mathcal{F}(x)))-V^{s-1}_\beta(\mathcal{S}(\mathcal{S}(x)))\right\} \\
&+(1-P(\bs{u}_2)) \left\{  1(\mathcal{F}(x)=1)\lambda_O + \beta V^{s-1}_\beta(\mathcal{F}(\mathcal{F}(x)))\right. \\
& \left. -V^{s-1}_\beta(\mathcal{F}(\mathcal{S}(x)))\right\}+C(\bs{u}_2)-C(\bs{u}_1)\\
&= P(\bs{u}_1) \left\{ \beta V^{s-1}_\beta(\mathcal{F}(\mathcal{S}(x)))-V^{s-1}_\beta(\mathcal{S}(\mathcal{S}(x)))\right\} \\
&+(1-P(\bs{u}_2)) \left\{ 1(\mathcal{F}(x)=1)\lambda_O + \beta V^{s-1}_\beta(\mathcal{F}(\mathcal{F}(x)))\right. \\
& \left. -V^{s-1}_\beta(\mathcal{S}(\mathcal{F}(x)))\right\}+C(\bs{u}_2)-C(\bs{u}_1).
\end{align*}
 \end{lemma}
\begin{lemma}\label{l3}
For $s=1,2,\ldots$, the functions $D^\beta_{s}(x)$ are decreasing in $x$ for $x\in \{1,2,\ldots,B-T+1\}$.
\end{lemma}
\begin{proof}
Within this proof, let $\pi_s^{\star}$ be the optimal policy for the $\beta$-discounted $s$ time-slots problem, and let $(\boldsymbol{u},\pi_{s-1}^{\star})$ be the policy for $s$ time-slots which takes the action $\boldsymbol{u}$ at the first time-slot, and then follows the policy $\pi_{s-1}^{\star}$. In order to prove the claim, we will use induction on $s$, the number of time-slots.

Let us assume that the statement is true for the functions $D^\beta_{z}(x)$, for all $z\leq s$. In particular the function, 
\begin{align}\label{assum1}
1(x=1)\lambda_O + \beta\left\{ V_\beta^{s-1}(\mathcal{F}(x))-V_\beta^{s-1}(\mathcal{S}(x))\right\}, 
\end{align}
is decreasing for $x\in \{1,2,\ldots,B-T+1\}$. 

First we will prove the decreasing property for $x\in \{2,3,\ldots,B-T+1\}$.
Now the assumption~\eqref{assum1} made above, and~\eqref{eq:1}, together imply that $\pi_s^{\star}$ is of threshold-type.

 Fix an $x \in \{1,2,\ldots,B-T\}$ and denote by $\boldsymbol{u}_1,\boldsymbol{u}_2,\boldsymbol{u}_3,\boldsymbol{u}_4$, the optimal actions at stage $s$ for the states $\mathcal{S}(x),\mathcal{F}(x),\mathcal{S}(x+1),\mathcal{F}(x+1)$ respectively. Note that the threshold nature of $\pi_s^{\star}$ implies that, 
\begin{align*}
&P(\boldsymbol{u}_1)<P(\boldsymbol{u}_2), P(\boldsymbol{u}_3)<P(\boldsymbol{u}_4)\mbox{ and },\\
&P(\boldsymbol{u}_3)<P(\boldsymbol{u}_1), P(\boldsymbol{u}_4)<P(\boldsymbol{u}_2).
\end{align*}
This is true because as the value of state decreases in the interval $\{1,2,\ldots,B\}$, a threshold policy switches to an action that has a higher transmission success probability. So it follows from Lemma~\ref{l2} that
\begin{align*}
&V_{\beta}^{s}(\mathcal{F}(x+1))-V_{\beta}^{s}(\mathcal{S}(x+1))\\
&\leq V_{\beta}^{s,(\boldsymbol{u}_2,\pi_{s-1}^{\star})}(\mathcal{F}(x+1))-V_{\beta}^{s}(\mathcal{S}(x+1))\\
& =C(\boldsymbol{u}_2) -C(\boldsymbol{u}_3) \\
&+ P_{c}(\boldsymbol{u}_3)\times \beta \left[V_\beta^{s-1}(\mathcal{F}(\mathcal{S}(x+1)))-V_\beta^{s-1}(\mathcal{S}(\mathcal{S}(x+1)))\right]\\
&+\left(1-P_{c}(\boldsymbol{u}_2)\right)\times \\
& \left\{1(\mathcal{F}(x+1)=1)+ \beta V_\beta^{s-1}(\mathcal{F}(\mathcal{F}(x+1))) \right. \\
&\qquad\qquad\left. - V_\beta^{s-1}(\mathcal{S}(\mathcal{F}(x+1))) \right\}\\
&\leq C(\boldsymbol{u}_2) -C(\boldsymbol{u}_3) \\
&+ P_{c}(\boldsymbol{u}_3)\times \beta \left[V_\beta^{s-1}(\mathcal{S}(\mathcal{F}(x)))-V_\beta^{s-1}(\mathcal{S}(\mathcal{S}(x)))\right]\\
&+\left(1-P_{c}(\boldsymbol{u}_2)\right)\times \\
&\left[1(\mathcal{F}(x)=1)+ \beta V_\beta^{s-1}(\mathcal{F}(\mathcal{F}(x))) - V_\beta^{s-1}(\mathcal{S}(\mathcal{F}(x)))\right]\\
&\leq V_{\beta}^{s}(\mathcal{F}(x))-V_{\beta}^{s}(\mathcal{S}(x)),
\end{align*}
where the first inequality follows since a sub-optimal action in the state $\mathcal{F}(x+1)$ increases the cost-to-go for $s$ time-slots, the second inequality is a consequence of the assumption that the functions $V_\beta^{s-1}(\mathcal{F}(x))-V_\beta^{s-1}(\mathcal{S}(x))$ are decreasing in $x$, while the last inequality follows from the fact that a sub-optimal action in the state $\mathcal{S}(x)$ will increase the cost-to-go for $s$ time-slots. Thus we have proved the decreasing property of $D^\beta_{s+1}(\cdot)$ for $x\in \{2,3,\ldots,B-T+1\}$, and it remains to show that $D^\beta_{s+1}(1)>D^\beta_{s+1}(2)$.

Once again, let $\boldsymbol{u}_1,\boldsymbol{u}_2,\boldsymbol{u}_3,\boldsymbol{u}_4$ be the optimal actions at stage $s$ for the states $T,0,T+1,1$ respectively. Using the same argument as above (i.e., assuming that the actions taken in stage $s$ at states $T,T+1$ are the same, and the actions taken in the states $0,1$ are the same), it follows that
\begin{align*}
& D_{s+1}(1)-D_{s+1}(2)\geq\\
& \qquad \left(1+\lambda_O -\beta \lambda_O \right)- \left(V^s_\beta(T)-V^s_\beta(T+1)\right).
\end{align*}
However, then $V^s_\beta(T)-V^s_\beta(T+1) \leq 1+\lambda_O -\beta \lambda_O$ (for $s$ stages, apply the same actions for the system starting in state $T$, as that for a system starting in state $T+1$, and note that the two systems couple at a stage $t-1$, when the latter system hits the state $1$ at any stage $t$; the hitting stage is of course random). This gives us,
\begin{align*}
& D_{s+1}(1)-D_{s+1}(2)\geq 0,
\end{align*}
and thus we conclude that the function $D_{s+1}(x)$ is decreasing for $x\in \{1,2,\ldots,B\}$. In order to complete the proof, we notice that for $s=1$, we have,
\begin{align*}
D_1^{\beta}(x) = 1(x=1)\lambda_O,
\end{align*}
and thus the assertion of Lemma is true for $s=1$.
\end{proof}
\begin{theorem}\label{t1}
Consider the single client problem discussed in Section~\ref{scp}. There is a threshold policy that is Blackwell optimal~\cite{blackwell}, i.e., it is optimal for all values of $\beta\in (\hat{\beta},1)$ for some $\hat{\beta}\in (0,1)$, and is also optimal for the Average cost problem. Thus $\pi_n^{\star}(\lambda_E)$ is of threshold-type and can be obtained in time $O(B^{E\times Q})$ via comparing the costs of all threshold-type policies.
\end{theorem}
\begin{proof}
Fix a $q$ and let $E_i,E_j, i>j$ be two power levels. Without loss of generality, let $\bs{u}_1=(q,E_i),\bs{u}_2=(q,E_j)$. Clearly $C(\bs{u}_1)>C(\bs{u}_2)$~\eqref{cost}. In the Bellman equation~\eqref{eq:1}, consider the term depending on $\bs{u}$, i.e. the term $C(\boldsymbol{u})-P(\boldsymbol{u}) D^\beta_s(x)$. For $x,y\in \{1,2,\ldots,B-T+1\}$, $x>y$, we have,
\begin{align*}
& C(\boldsymbol{u}_1)-P(\boldsymbol{u}_1) D^\beta_s(x) - \left(C(\boldsymbol{u}_2)-P(\boldsymbol{u}_2) D^\beta_s(x)\right)\\
&-\{C(\boldsymbol{u}_1)-P(\boldsymbol{u}_1) D^\beta_s(y) - \left(C(\boldsymbol{u}_2)-P(\boldsymbol{u}_2) D^\beta_s(y)\right)\}\\
&=\left(P(\boldsymbol{u}_1)  - P(\boldsymbol{u}_2)\right) \left(D^\beta_s(y)-D^\beta_s(x)\right)\\
&\geq 0,
\end{align*}
where the last inequality follows from Lemma~\ref{l3}. Thus it follows that if action $\bs{u}_1$ is preferred over action $\bs{u}_2$ for any state $x$, then $\bs{u}_1$ will also be preferred over action $\bs{u}_2$ for any state $y<x$, $y\in \{1,2,\ldots,B-T+1\}$. Finally note that it follows from the Bellman equation~\eqref{eq:1} and~\eqref{sf}, that the optimal action for states $x>B-T+1$ is to let $E=0$ (since any packet that is received will be lost due to buffer over flow). The proof for variations in power levels is similar. Thus it follows from the definition of a threshold policy that the optimal policy is of threshold type.

Finally note that the statement regarding Blackwell optimality follows from the result in the above paragraph, and because the state-space is finite.
\end{proof}
\section{Solution of Primal MDP}
We now present the solution of the Primal Problem.
\begin{lemma}\label{l4}
$D(\lambda_E) = \sum_n V_n(\lambda_E)-\lambda_E \bar{E}$.
\end{lemma}
\begin{proof}
Let $\pi^{\star}(\lambda_E):= \otimes \pi_n^{\star}(\lambda_E)$ be the policy obtained by following the policy $\pi_n^{\star}(\lambda_E)$ for each client $n$. Then from the definition of dual function, Lagrangian~\eqref{lagrange}, cost associated with a policy $\pi$~\eqref{scopol} and Lemma~\ref{l1}, we have
\begin{align}\label{eq:5}
\mathcal{L}(\pi,\lambda_E) \geq D(\lambda_E)  \geq \sum_n V_n(\lambda_E,\pi) -\lambda_E \times \bar{E}.
\end{align}
However since the policy $\pi^{\star}(\lambda_E)$ is stationary, (all the $\liminf$ and $\limsup$  become $\lim$ in the definition of its Lagrangian, and associated rewards in the single-client problem change to $\lim$), we have that
\begin{align*}
\mathcal{L}(\pi^{\star}(\lambda_E),\lambda_E) = \sum_n V_n(\lambda_E) -\lambda_E \times \bar{E},
\end{align*}
which, along with~\eqref{eq:5} gives us $D(\lambda_E) = \sum_n V_n(\lambda_E)-\lambda_E \bar{E}$. 
\end{proof}

\begin{theorem}\label{t2}
Consider the Primal MDP~\eqref{pmdp} and its associated dual problem defined in~\eqref{dual}. There exists a price $\lambda^\star_E$ such that $\left(\pi^{\star}(\lambda^\star_E),\lambda^\star_E\right)$ is an optimal primal-dual pair and thus the policy $\pi^{\star}(\lambda^\star_E)$ solves the Primal MDP. 
\end{theorem}
\begin{proof}
We observe that there is a one-to-one correspondence between any stationary randomized policy, and the measure it induces on the state-action space, and thus the Primal MDP can be posed as a linear program~\cite{ergodic1,ergodic2}. Thus it follows from Slater's condition~\cite{Ber87} that for the Primal MDP, strong duality holds if there exists a policy $\pi$ that satisfies the constraints $\limsup_{t\to\infty}\frac{1}{t}\mathbb{E}\sum_n\sum_{s}E_n(s) < \bar{E}$. However the policy which never schedules any packets incurs a net power expenditure of $0$, and thus Slater's condition is true for the Primal MDP if $\bar{E}>0$. The claim of the Theorem then follows from Lemma~\ref{l3}.
\end{proof}
We note that the policy $\pi^{\star}(\lambda^\star_E)$ is a decentralized policy. That is, the decision to choose the video-quality and power-level at each time $t$ for client $n$, i.e., $\left(q_n(t),E_n(t)\right)$ can be taken by client $n$ itself, and doesn't require the AP to co-ordinate the clients. Thus a client $n$ need not know the state values of other clients, $l_m(t)$ for $m\neq n$, nor does the AP need to know the values of $l_n(t)$. Thus the policy is easy to implement.
\subsection{Obtaining $\lambda^{\star}_E$ iteratively in a decentralized fashion}
We note that in order to implement the optimal policy $\pi^{\star}(\lambda^\star_E)$ as in Theorem~\ref{t2}, we need to find the optimal value of the price $\lambda^\star_E$. We iterate on the price $\lambda_E$ using the sub-gradient method~\cite{shor}, and since the problem is concave, the prices converge to the optimal value $\lambda^\star_E$. Moreover the iterations involving price-updates are decentralized, i.e., the clients need only the knowledge of the current price $\lambda_E$ for the iteration.

Now since $ D(\lambda_E)  = \mathcal{L}(\pi^\star(\lambda_E),\lambda_E)$, we have,
\begin{align}\label{partial1}
\frac{\partial D}{\partial \hat{\lambda}_v} =\bar{E}- \mathbb{E}_{\pi^\star(\lambda_E)} \sum_n \tau(n,\pi^\star(\lambda_E)),
\end{align} 
where $\mathbb{E}_{\pi^\star(\lambda_E)} \sum_n \tau(n,\pi^\star(\lambda_E))$ is the expected cost incurred on the power over all the users. This is the total ``congestion" at the AP.
The iteration for $\lambda_E$ is,
\begin{align*}
\lambda^{k+1}_E = \lambda^k_E - \alpha_k g_k,
\end{align*}
where $d_k$ is the sub-gradient evaluated in~\eqref{partial1}.
\section{Fading Channels}\label{fading}
The results in the previous sections can be extended in a straight forward manner to the case of fading channels. Let the channel conditions for client $n$ be described by a Markov process evolving on finitely many states $\{1,2,\ldots,C_n\}$ having a transition matrix $\Pi_n$. The state of client $n$ is described by the vector $\boldsymbol{x}_n(t):=\left(l_n(t),c_n(t)\right)$, where $l_n(t)$ is the play-time duration of the packets present in the buffer at time $t$, and $c_n(t)$ is the channel condition at time $t$. 
If the client $n$ is scheduled a packet transmission of quality $q$ at an power $E$ at time $t$, then the system state at time $t+1$ is $\left(\mathcal{S}(l(t)),\tilde{c}\right)$ with a probability $P_{n,c_n(t)}(q,E)\Pi(c_n(t),\tilde{c})$, while it is $\left(\mathcal{F}(l(t)),\tilde{c}\right)$ with a probability $P_{n,c_n(t)}(q,E) \Pi(c_n(t),\tilde{c})$.

However now the cost associated to an action $\bs{u}$ also depends on the channel condition, i.e.,
\begin{align}\label{cost1}
C_c(\boldsymbol{u}):=\lambda_E E + P_{c}(l,E)\lambda_q ,
\end{align}
and a threshold policy will have a threshold structure for each value of channel condition (as defined in Section~\ref{scp}). 
\section{Concluding Remarks}
We have formulated the problem of dynamically choosing the qualities and power levels for packet transmissions across unreliable wireless so as to maximize the Quality of Experience of video streaming channels as an MDP. Using Lagrangian techniques, we have shown that the problem exhibits a decentralized solution, wherein the clients can dynamically decide these quantities on their own using their local information, i.e., the channel state and the amount of playtime remaining in their buffers. Thus the optimal policy can be obtained in time linear in the number of users.

Furthermore we have shown that the optimal policy has a threshold structure, thus further reducing the complexity of searching for the optimal policy. Moreover due to the threshold nature of the policy, it is easy to implement.

\end{document}